\documentclass[runningheads,envcountsame]{llncs}

\usepackage[utf8]{inputenc}
\usepackage[T1]{fontenc}
\usepackage{graphicx} 
\usepackage{amsfonts, amsmath, amssymb, mathrsfs}
\usepackage{color}
\usepackage{tikz}
\usepackage{contour}
\usepackage{bbold}
\usepackage{bbding}
\usepackage{colonequals}
\usepackage{etoolbox}
\usepackage[noadjust]{cite}
\usepackage[noxy]{virginialake}
\usepackage{bussproofs} \EnableBpAbbreviations
\usepackage{booktabs}
\usepackage{relsize}
\usepackage{adjustbox}
\usepackage{multirow}
\usepackage{paralist}
\usepackage{url}

\usepackage{hyperref}
\usepackage[capitalise]{cleveref}
\crefname{section}{Sect.}{Sects.}
\crefname{appendix}{App.}{Apps.}
\crefname{definition}{Def.}{Defs.}
\crefname{proposition}{Prop.}{Props.}
\Crefname{section}{Section}{Sections}
\Crefname{appendix}{Appendix}{Appendices}
\Crefname{definition}{Definition}{Definitions}
\Crefname{proposition}{Proposition}{Propositions}

\catcode`@=11
\def\Left#1#2\Right{\begingroup%
   \def\ts@r{\nulldelimiterspace=0pt \mathsurround=0pt}%
   \let\@hat=#1%
   \def\sht@im{#2}%
   \def\@t{{\mathchoice{\def\@fen{\displaystyle}\k@fel}%
          {\def\@fen{\textstyle}\k@fel}%
          {\def\@fen{\scriptstyle}\k@fel}%
          {\def\@fen{\scriptscriptstyle}\k@fel}}}%
   \def\g@rin{\ts@r\left\@hat\vphantom{\sht@im}\right.}%
   \def\k@fel{\setbox0=\hbox{$\@fen\g@rin$}\hbox{%
      $\@fen \kern.4\wd0 \copy0 \kern-.4\wd0%
      \llap{\copy0}\kern.4\wd0$}}%
      \def\pt@h{\mathopen\@t}\pt@h\sht@im%
      \Right}%
\def\Right#1{\let\@hat=#1%
   \def\st@m{\mathclose\@t}%
   \st@m\endgroup}
\catcode`@=12

\newcommand{\K}{{\sf K}}
\newcommand{\B}{{\sf B}}
\newcommand{\Lgk}{{\sf L}}
\newcommand{\Sfive}{{\sf S5}}

\newcommand{\CSfive}{\sharp{\sf S5}_n}

\newcommand{\TCSfive}{\overline{\sharp{\sf S5}}_n}

\newcommand{\ce}{\colonequals}
\newcommand{\cce}{\coloncolonequals}

\newcommand{\prop}{\mathrm{Prop}}
\newcommand{\agt}{\mathrm{Agt}}

\newcommand{\vlfill}[1]{\{#1\}}
\renewcommand{\vlhole}{\vlfill{\:}}
\newcommand{\id}{\mathrm{id}}
\newcommand{\hasbracket}{\nearrow}
\newcommand{\hasparent}{\searrow}
\newcommand{\nobracket}{\not\rightsquigarrow}

\newcommand{\seqint}{\Left[\cdot\Right]}
\newcommand{\labint}[1]{\Left[#1\Right]}
\newcommand{\onlyb}{\aleph}
\newcommand{\La}{\mathbb{L}}

\newcommand{\Var}{\textup{\texttt{Var}}}
\newcommand{\Varp}{{\textup{\texttt{Var}}^+}}
\newcommand{\Varn}{{\textup{\texttt{Var}}^-}}
\newcommand{\Varm}{{\textup{\texttt{Var}}^*}}
\newcommand{\Ag}{\textup{\texttt{Ag}}}
\newcommand{\Col}{\textup{\texttt{Col}}}
\newcommand{\Colp}{{\textup{\texttt{Col}}^+}}
\newcommand{\Coln}{{\textup{\texttt{Col}}^-}}
\newcommand{\Colm}{{\textup{\texttt{Col}}^*}}

\newcommand{\fd}{\mathit{fd}}

\newcommand{\labe}[2]{{#1}\mathord{:}{#2}}
\newcommand{\spdisj}{\mathbin{\varovee}}
\newcommand{\spconj}{\mathbin{\varowedge}}
\newcommand{\bigspdisj}{\mathop{\mathlarger{\mathlarger{\mathlarger{\mathlarger{\varovee}}}}}\limits}
\newcommand{\bigspconj}{\mathop{\mathlarger{\mathlarger{\mathlarger{\mathlarger{\varowedge}}}}}\limits}
\newcommand{\rsf}{\mathsf{r}}
\newcommand{\interp}[2]{#1\xrightarrow{\textup{\prp{\tiny IP}}}#2}

\newcommand{\LIPinterp}[2]{#1\xrightarrow{\textup{\prp{\tiny LIP}}}#2}
\newcommand{\ALIPinterp}[2]{#1\xrightarrow{\textup{\prp{\tiny ALIP}}}#2}

\newcommand{\prp}[1]{{\sf #1}}

\newcommand{\notp}{\neg p}
\renewcommand{\phi}{\varphi}
\newcommand{\M}{\mathcal{M}}

\newcommand{\CS}{\mathcal{S}}
\newcommand{\CT}{\mathcal{T}}
\newcommand{\CU}{\mathcal{U}}

\newcommand\newthanks[1]{%
  \begingroup
  \begin{NoHyper}
  \renewcommand\thefootnote{}\footnote{#1}%
  \end{NoHyper}
  \addtocounter{footnote}{-1}%
  \endgroup
}

\begin{document}
 \title{\texorpdfstring{Agent Interpolation for  Knowledge}{Agent Interpolation for  Knowledge}}
 \titlerunning{Agent Interpolation for  Knowledge}
\author{Marta B\'{\i}lkov\'a\orcidID{0000-0002-3490-2083} \and  
Wesley Fussner\orcidID{0000-0002-6025-8770} \and  
Roman Kuznets\orcidID{0000-0001-5894-8724}\Envelope}
\authorrunning{M. B\'\i{}lkov\'a \and W. Fussner \and R. Kuznets}
\institute{Institute of Computer Science of the Czech Academy of Sciences, \email{\{bilkova,fussner,kuznets\}@cs.cas.cz}}

\maketitle
\begin{abstract}
We define a new type of proof formalism for multi-agent modal logics with S5-type modalities. This novel formalism combines the features of hypersequents to represent S5 modalities with nested sequents to represent the T-like modality alternations. We show that the calculus is sound and complete, cut-free, and terminating and yields  decidability and the finite model property for multi-agent S5. We also use it to prove the Lyndon (and hence Craig) interpolation property for multi-agent S5, considering not only propositional atoms but also agents to be part of the common language. Finally, we discuss the difficulties on the way to extending these results to the logic of distributed knowledge and to deductive interpolation. 

\keywords{multi-agent systems 
\and interpolation \and modal logic}
\end{abstract}

\section{Introduction}\newthanks{The authors were supported by the grant 22-23022L CELIA of Grantov\'{a} Agentura \v{C}esk\'{e} Republiky.}
\label{section:introduction}
One logic often has multiple proof-theoretic representations. In most cases, it is not reasonable to see one of them as being objectively better than all others in all respects. A Hilbert-style calculus could be good for proving semantical completeness but give little help for establishing decidability. A cut-free sequent calculus works much better for decidability but may be unsuitable to show the Craig interpolation property (\prp{CIP}) even if it does hold (cf.,~e.g.,~\cite{Dyckhoff99} for the G\"odel--Dummett logic). If no cut-free calculus exists, an analytic sequent calculus could be one of the simplest alternatives and may even be sufficient for establishing the \prp{CIP}, but the Lyndon interpolation property (\prp{LIP}) would require something else, e.g.,~hypersequents. One-size-fits-all proof calculi are exceedingly rare.

The choice of calculus is, therefore, dictated, on the one side, by applications and, on the other side, by the range of logics to be applied to. Quite often, the choices can be mapped onto a line between generality and efficiency. The more logics a proof formalism covers, the less likely it is to be efficient for each one of them. Labelled sequents~\cite{Negri05} are on the more general part of the spectrum, with the methods described in~\cite{NegrivP11} surely sufficient to treat a logic like multi-agent~$\Sfive_n$. The goal of this work is to develop a proof formalism tailor-made for~$\Sfive_n$ and extract from it a number of results such as decidability and interpolation for both propositional atoms (for which we obtain Lyndon interpolation) and agents (for which we obtain Craig interpolation). We also discuss interpolation for distributed knowledge and agent interpolation for the global consequence relation, which prove more subtle.

The format we develop is a hybrid of \emph{hypersequents}~\cite{Minc71,Pottinger83,Avron96} to represent equivalence classes of individual agents and \emph{nested sequents}~\cite{Bruennler09,Poggiolesi09,LellmannP24} to represent the tree structure for alternating agent modalities. Since ``nested hypersequents'' could be confused with nested sequents, which are sometimes called \emph{tree-hypersequents}~\cite{Poggiolesi09}, we will use a fresh term \emph{crossword sequents} because our proof structures resemble (multidimensional) crosswords, with multiple hypersequents intersecting in a sequent component, shortening it to \emph{cross-sequents} because they are a cross between hypersequents and nested sequents.

\section{Preliminaries}
\label{sect:formalism}

We work with the multi-agent modal language in the negation normal form~(NNF). The latter choice enables us to use one-sided sequents. This choice is inessential and can be easily  changed to the full language with negation and two-sided sequents. Formulas are   defined by the grammar $\phi \cce \bot\mid \top\mid p \mid \notp \mid (\phi \wedge \phi) \mid (\phi \vee \phi) \mid \Box_a \phi \mid \Diamond_a \phi$ where $p \in \prop$ is a (propositional) atom, $\notp$ is the negation of $p$, with both called \emph{literals}, and $a \in A$ is an agent for some countably infinite $\prop$ and finite  $A\ne \varnothing$. The total number of agents is~$n$. Negation $\neg$ can be recursively extended from  atoms to arbitrary formulas using standard De~Morgan dualities, e.g., $\neg\Box_a \phi \ce \Diamond_a \neg\phi$. Implication $\phi \to \psi \ce \neg\phi\lor \psi$. The variables and agents occurring (positively/negatively) in a formula are defined in the standard way, with $\Col \in \{\Var, \Ag\}$, $\circ\in \{\wedge, \vee\}$, $\ast\in\{+,-\}$,  $\heartsuit \in\{\Box, \Diamond\}$,  and $\Col(\phi)\ce \Colp(\phi) \cup \Coln(\phi)$: $\Var(\bot)=\Var(\top)=\Ag(\bot)=\Ag(\top)=\Ag(p)=\Ag(\notp)\ce\varnothing$, so that in these cases the sets $\Varp$ and $\Varn$ of positively and negatively occurring variables  are also empty; further, $\Colm(\phi \circ \psi) \ce \Colm(\phi) \cup \Colm(\psi)$; $\Varm(\heartsuit_a \phi) \ce \Varm(\phi)$; $\Varp(p)=\Varn(\notp)\ce\{p\}$; $\Varn(p)=\Varp(\notp)\ce\varnothing$; $\Ag(\heartsuit_a \phi)\ce \{a\} \cup \Ag(\phi)$. Modal depth $\fd(\phi)$ is also defined recursively: $\fd(\top)=\fd(\bot)=\fd(p)=\fd(\notp)\ce0$, $\fd(\phi\circ\psi)\ce \max(\fd(\phi),\fd(\psi))$, and $\fd(\heartsuit_a \phi)\ce \fd(\phi)+1$. 

Modal logic~$\Sfive_n$ is the logic of multi-agent Kripke frames with all  accessibility relations being equivalence relations. 
\begin{definition}[Semantics]
A structure $\M=(W,R,V)$ is an \emph{(epistemic Kripke) model}
if{f} $W \ne \varnothing$ is a set of worlds, $R$ maps each agent $a \in A$ to an equivalence relation\footnote{An \emph{equivalence relation} $R$ on $W$ is reflexive, symmetric, and transitive, i.e., $w R w$, and $v R w$ whenever $w R v$, and $w R u$ whenever $w R v$ and $v R u$ for all $w,v,u \in W$.} $R_a \subseteq W \times W$ on $W$, and $V$ maps each  atom $p\in \prop$ to a set $V(p) \subseteq W$ of worlds where $p$ is true.  $R_a(w) \ce \{ v 
\in W \mid w R_a v\}$.
Each $R_a$ partitions $W$ into pairwise disjoint \emph{clusters} $\{R_a(w) \mid w \in W\}$, where an $a$-cluster is a maximal totally $R_a$-connected set $U\subseteq W$, i.e., $(u,v)\in R_a$ for all $u,v \in U$ and  $(u,w)\notin R_a$ for all $u \in U$ and $w \notin U$. $uR_av$ means $(u,v) \in R_a$.

Truth of a formula $\phi$ in a world $w$ of a model $\M$ is defined recursively: 
$\M, w \nvDash \bot$; $\M, w \vDash \top$;  $\M, w \vDash p$ if{f} $w \in V(p)$; $\M, w \vDash \notp$ if{f} $w \notin V(p)$; $\M, w \vDash \phi_1 \wedge \phi_2$ if{f} $\M, w 
\vDash \phi_i$ for both $i=1,2$; $\M, w \vDash \phi_1 \vee \phi_2$ if{f} $\M, w 
\vDash \phi_i$ for at least one $i=1,2$; $\M,w\vDash \Box_a \phi$ if{f} $\M, v \vDash \phi$ for all $v \in R_a(w)$; and $\M,w\vDash \Diamond_a \phi$ if{f} $\M, v \vDash \phi$ for at least one $v \in R_a(w)$. A formula $\phi$ is $\Sfive_n$-\emph{valid}, written $\phi \in \Sfive_n$ if{f} it is true in all worlds of all epistemic models.
\end{definition}

The Hilbert-style axiomatization of $\Sfive_n$ via factivity and positive and negative introspection is well known (see, e.g.,~\cite{FaginHMV95}). 
It is axiomatized by any complete set of axioms for the classical propositional logic, rule \emph{modus ponens}, and the following axioms and rules for each agent $a \in A$: 
 \begin{gather*}
     \Box_a (\phi \to \psi) \to (\Box_a \phi \to \Box_a \psi)
     \qquad
     \Box_a \phi \to \phi
     \qquad
     \frac{\phi}{\Box_a \phi}
     \\
     \Box_a \phi \to \Box_a\Box_a \phi
     \qquad
     \Diamond_a \phi \to \Box_a\Diamond_a \phi
 \end{gather*}

 \begin{theorem}[Semantic Completeness]
     $\vdash_{\Sfive_n} \phi$ if{f} $ \phi \in \Sfive_n$.
 \end{theorem}

The goal of this paper is to provide a proof-theoretic description of~$\Sfive_n$.

A \emph{(one-sided)  sequent} is a finite set of formulas, written $\Delta=\phi_1,\dots,\phi_k$. It is  true at a world $w$ of model $\M$, written $\M, w \vDash \Delta$, if{f} $\M, w \vDash \phi_i$ for some~$\phi_i$, i.e., there is a formula interpretation $\iota_s(\Delta) \ce \bigvee\Delta=\bigvee_{i=1}^k\phi_i$ such that $\M, w \vDash \Delta$ if{f} $\M, w \vDash \iota_s(\Delta)$. No cut-free \emph{sequent} calculus has been found for the single-agent~$\Sfive_1$. One way to avoid using the cut rule is to employ  \emph{hypersequents}~\cite{Avron96}, which  are multisets of sequents. A hypersequent $H=\Delta_1 \mid \dots \mid \Delta_m$  has the formula interpretation $\iota_h(H)\ce \bigvee_{j=1}^m \Box \iota_s(\Delta_j)$. Thus, falsifying a hypersequent means finding a cluster containing worlds $w_1,\dots,w_m$ such that $\M, w_j \nvDash \Delta_j$ for each $j=1,\dots,m$.

The multiset sequent structure is sufficient  for $\Sfive_1$ because it matches the accessibility structure of its models, but many modal logics require more structure. For instance, \emph{nested sequents}~\cite{LellmannP24}, which are trees of sequents, are not as general as labelled sequents, but provide a robust and purely syntactical representation for many logics of note, including  $\K_n$. The accessibility relations for $\K_n$ have  no additional properties, which is relevant for $\Sfive_n$ because its alternating modalities exhibit certain $\K$-like properties. For instance, the common explanation for the  complexity jump from (co-)NP-complete for $\Sfive_1$ to PSPACE-compelete for $\Sfive_n$ for $n>1$~\cite[Table 3.2]{FaginHMV95} is that logic $\K$ is PSPACE-complete and the modality combination $\Box_1\Box_2$  behaves like a $\K$-modality.  

Combining these intuitions, we propose a tree of hypersequents to be the simplest data structure matching the accessibility structure of Kripke models for  $\Sfive_n$. There is  one more puzzle piece to discuss before giving  formal definitions. Logic~$\K_n$ can be captured by sequents without resorting to the nested tree structure. The corresponding tableaus are called \emph{destructive}~\cite{HandbookTM99} in the sense that information can be lost. Intuitively, after exploring one world, a non-invertible sequent rule is applied to move to an accessible world and never look back. Not looking back would, however, seem problematic for symmetric accessibility relations of~$\Sfive_n$, which may explain why no cut-free sequent calculus is known to capture symmetric logics like $\Sfive$ or $\B$ (even for a single agent). In the nested hypersequent format considered here all rules will be invertible, facilitating the  flow of information from the leaves back to the root. 

\section{Crossword Sequents}
\label{sect:crossword}
\emph{Crossword sequents} $\CS$, or \emph{cross-sequents} for short, their depth~$d(\CS)$, and formula interpretation $\iota(\CS)$ are defined recursively, where $\Delta$ is a sequent, usually called \emph{(sequent) component}, and $H$ is a \emph{hypersequent}, understood here as a multiset of cross-sequents:
\begin{compactitem}
    \item a sequent $\Delta$ is a cross-sequent with $d(\Delta)=0$ and  $\iota(\Delta)\ce \iota_s(\Delta)$;
    \item if $a \in A$ is an agent and $H=\CS_1|\dots|\CS_m$ is a hypersequent, then $\CS, [H]_a$ is a cross-sequent with $d(\CS, [H]_a) \ce \max(d(\CS), d(H)+1)$ and $\iota(\CS, [H]_a) \ce \iota(\CS) \vee \bigvee_{j=1}^m \Box_a\iota(\CS_j)$;
    \item a cross-sequent $\CS$ is a (singleton) hypersequent with depth $d(\CS)$;
    \item  $H \mid \CS$ is a  hypersequent of depth $d(H | \CS) \ce \max(d(H),d(\CS))$.
\end{compactitem}
(We do not define formula interpretations for hypersequents as it depends on the agent index of the surrounding bracket.)
Thus, each cross-sequent $\CS$ has a unique component as its root $r(\CS)$, whereas a  hypersequent $H$ has several roots. Each component $\Delta$ of $\CS$ other than the root $r(\CS)$ belongs to an agent-indexed bracket. The root $r(\CS, [H]_a)$ is the \emph{$a$-parent} of all roots of $H$, which are its \emph{$a$-children}, and are \emph{$a$-siblings} among themselves. $\CS \hasbracket a$ denotes that $r(\CS)$~has an $a$-child, otherwise $\CS \nobracket a$. The parent/child/sibling relationships among components are preserved when embedded into a larger cross-sequent. For a  component~$\Omega$, we write $\Omega \in a$ if{f}  $\Omega$ has an $a$-parent and $\Omega \notin a$ otherwise. For bookkeeping purposes, we will also label individual components by unique non-negative integers. In such cases, we write $\CS_\ell$ for the cross-sequent~$\CS$ supplied with  the labeling~$\ell$ and  denote the label (name) of $\Delta$ by $\ell(\Delta)$, with $\ell(r(\CS))$ always set to~$0$.\footnote{One could use labels with structure reflecting the tree structure of cross-sequents, cf., e.g., sequences of integers used in prefixed tableaus to represent nested sequents~\cite{Fitting12}.  We opted for arbitrary labels here to avoid cumbersome label management. } Distinct metavariables $\Delta_1$ and $\Delta_2$ over sets are used even when $\Delta_1=\Delta_2$ as sets of formulas, making the notation $\ell(\Delta)$ unambiguous within the same cross-sequent. $\La(\CS)$ denotes the set of all label of the components of~$\CS$. $\fd(\CS)$ is defined as the maximum over $\fd(\phi)$ over all formulas $\phi$ occurring in components of $\CS$.
\begin{example}
    The following cross-sequent $\CS$ consists of 10 sequent components:
    \begin{equation}
    \label{ex:cross}
       \CS= \Delta, \Bigl[\Gamma_1, [\Theta_1|\Theta_2|\Theta_3]_b  \Big| \Gamma_2\Bigr]_a, \Bigl[\Lambda, \bigl[\Xi_1|\Xi_2, [\Pi]_a\bigr]_c\Bigr]_b,
    \end{equation}
    which could be labeled by numbers from 0 to 9.
    Component $\Delta$ is the root $r(\CS)$ with $\ell(\Delta)=0$. $\Xi_2$ is the $a$-parent of $\Pi$, and  $\Gamma_1$ is the $b$-parent of  $\Theta_1$, $\Theta_2$, and~$\Theta_3$. Note that neither $\Xi_1$  nor $\Lambda$ is a parent of $\Pi$.  $\Theta_1$, $\Theta_2$, and $\Theta_3$ are $b$-siblings and each of them is a $b$-child of $\Gamma_1$. $\CS \hasbracket a$ and $\CS \hasbracket b$, but $\CS \nobracket c$. 
    $\Pi \in a$, $\Xi_1 \in c$, and $\Lambda \in b$, but $\Pi \notin c$. 
\end{example}
\begin{definition}[Consolidation conditions] A cross-sequent $\CS$ is \emph{proper} if{f} the following \emph{consolidation conditions} are satisfied:
\begin{compactenum}[(a)]
\item\label{cond:noasiblings} two $a$-children of the same component must be $a$-siblings, i.e., one  component cannot contain more than one $a$-bracket per agent $a$; 
\item\label{cond:noachildren} an $a$-child cannot be an $a$-parent, i.e., a  component within an $a$-bracket itself cannot contain $a$-brackets. 
\end{compactenum}
\end{definition}
For instance, $\Delta, [\CS]_a, [\CS']_a$ and $\Delta, [\CS, [\CS']_a]_a$ are not proper cross-sequents:  one should use $\Delta, [\CS | \CS']_a$ instead. From now on all cross-sequents are assumed to be proper, as was~\eqref{ex:cross}.

Like worlds in a model, we combine components in a cross-sequent into \emph{component clusters}. They will play a role in the proof calculus.
\begin{definition}[Component cluster]
In a cross-sequent~$\CS$ (with component labeling $\ell$), component $\Delta$ belongs to $[\Gamma]_a$, the \emph{$a$-component cluster} of component~$\Gamma$,  if{f} (a)~$\ell(\Delta)=\ell(\Gamma)$, i.e., $\Delta$ is $\Gamma$ itself, or $\Delta$ is (b)~its $a$-parent, or (c)~its $a$-sibling, or (d)~its $a$-child. All components of $\CS$ are partitioned into $a$-clusters~$[\Gamma]_a$. An  $a$-cluster consisting of exactly one component is  \emph{trivial}. Thus, a \emph{non-trivial} $a$-cluster has at least two components, one of which is the $a$-parent of all other components in the cluster and called the \emph{$a$-parent of the cluster}.
\end{definition}

\begin{example}[Clusters in cross-sequents]
The clusters from~\eqref{ex:cross} are listed below with non-trivial clusters boldfaced. The $a$-clusters are $\mathbf{\{\Delta, \Gamma_1,\Gamma_2\}}$, $\{\Theta_1\}$, $\{\Theta_2\}$, $\{\Theta_3\}$, $\{\Lambda\}$, $\{\Xi_1\}$, and $\mathbf{\{\Xi_2, \Pi\}}$. The $b$-clusters are $\mathbf{\{\Delta,\Lambda\}}$, $\mathbf{\{\Gamma_1,\Theta_1,\Theta_2,\Theta_3\}}$, $\{\Gamma_2\}$, $\{\Xi_1\}$,  $\{\Xi_2\}$, and $\{\Pi\}$. Finally, the $c$-clusters are
$\{\Delta\}$, $\{\Gamma_1\}$, $\{\Theta_1\}$, $\{\Theta_2\}$, $\{\Theta_3\}$, $\{\Gamma_2\}$,  $\mathbf{\{\Lambda,\Xi_1,\Xi_2\}}$, and $\{\Pi\}$. 
The crossword analogy comes from the fact that $[\Sigma]_d \cap [\Sigma]_e =\{\Sigma\}$ for any component $\Sigma$ and $d \ne e$. 
Note that more than two non-trivial clusters can intersect in one component.
\end{example}

The rules of our cross-sequent calculus $\CSfive$ can be found in \cref{fig:S5n}. They are written in the standard nested sequent fashion with $\CS\vlfill{X}$ representing the result of inserting $X$ in place of the the \emph{hole} $\vlhole$ in $\CS\vlhole$.

\begin{figure}[t]
\centering{
\fbox{\parbox{0.98\textwidth}{
\[
\vlinf{\id}{}{\CS\vlfill{p,\notp}}{}
\qquad
\vlinf{\top}{}{\CS\vlfill{\top}}{}
\]
\[
\vlinf{\vee}{}{\CS\vlfill{\phi\vee\psi}}{\CS\vlfill{\phi\vee\psi,\phi,\psi}}{} 
\qquad
\vliinf{\wedge}{}{\CS\vlfill{\phi\wedge\psi}}{\CS\vlfill{\phi\wedge\psi,\phi}}{\CS\vlfill{\phi\wedge\psi,\psi}}{}
\]
\[
\vlinf{\Box \in a}{}{\CS\vlfill{[\Box_a \phi, \CT|H]_a}}{\CS\vlfill{[\phi|\Box_a \phi, \CT|H]_a}}{} 
\qquad
\vlinf{\Box \hasbracket a}{}{\CS\vlfill{\Box_a \phi, [H]_a}}{\CS\vlfill{\Box_a \phi, [\phi|H]_a}}{} 
\qquad
\vlinf{\Box\nobracket a}{\scriptscriptstyle \Box_a \phi \nobracket a}{\CS\vlfill{\Box_a \phi}}{\CS\vlfill{\Box_a \phi, [\phi]_a}}{} 
\]
\[
\vlinf{\Diamond \in a}{}{\CS\vlfill{[\Diamond_a \phi, \CT|\CU|H]_a}}{\CS\vlfill{[\Diamond_a \phi, \CT|\phi,\CU|H]_a}}{} 
\qquad
\vlinf{\Diamond \hasbracket a}{}{\CS\vlfill{\Diamond_a \phi, [\CT|H]_a}}{\CS\vlfill{\Diamond_a \phi, [\phi,\CT|H]_a}}{} 
\]
\[
\vlinf{\Diamond T a}{}{\CS\vlfill{\Diamond_a \phi}}{\CS\vlfill{\Diamond_a \phi,\phi}}{} 
\qquad
\vlinf{\Diamond \hasparent a}{}{\CS\vlfill{\CT,[\Diamond_a \phi,\CU|H]_a}}{\CS\vlfill{\phi,\CT,[\Diamond_a \phi,\CU|H]_a}}{} 
\]
}}}
\caption{Calculus $\CSfive$. All cross-sequents are proper. Only one of the rules $\Box \in a$, $\Box \hasbracket a$, and $\Box\nobracket a$ can be applied (in the bottom-up direction)  to any occurrence of $\Box_a \phi$, depending on its sequent component: $\Box \in a$ creates another $a$-sibling if the component already had $a$-siblings; $\Box \hasbracket a$ creates another $a$-child if the component already had $a$-children; $\Box\nobracket a$ creates the first $a$-child otherwise. The consolidation conditions make sure that the three situations are mutually exclusive. The rules $\Diamond \in a$, $\Diamond \hasbracket a$, $\Diamond T a$, and $\Diamond \hasparent a$, taken together, ensure that $\phi$ can be added to every sequent component of the $a$-cluster containing $\Diamond_a \phi$.}
\label{fig:S5n}
\end{figure}

\begin{definition}[Proof tree labeling]
\label{def:proof_label}
In the process of proof search, while constructing  proofs from derivations, i.e., from partial proofs, bottom up in~$\CSfive$, one consolidated labeling $\ell$ is used for all proof trees, imposing the \emph{coherency conditions} that (A)~applying cross-sequent rules upward, all components retain their labels, and (B)~newly created components are given fresh labels, not occurring in the whole proof tree up to this point.%
\end{definition}

Intuitively, clusters of components in a cross-sequent  represent clusters of worlds in a model. Accordingly, the modal rules for~$\Box_a$~and~$\Diamond_a$ will treat all components from an $a$-cluster as being mutually accessible. For instance, if $\Box_a \phi \in \Gamma$ and $[\Gamma]_a$ is non-trivial, then  $\Box_a \phi$ would create an additional $a$-child of the $a$-parent of $[\Gamma]_a$, which would become an $a$-sibling of all other existing components in the cluster as in hypersequents. If, however, the cluster $[\Gamma]_a$ is  trivial, then the first $a$-child of $\Gamma$ is created, turning $\Gamma$ into an $a$-parent as in nested sequents. Note also that $\Diamond_a \phi\in\Gamma$ should have access to the $a$-parent of a non-trivial $[\Gamma]_a$, meaning that formulas can travel towards the root.

\begin{definition}[Cross-sequent semantics]
An \emph{interpretation} $\seqint$ of a labeled cross-sequent $\CS_\ell$  into a model $\M=(W,R,V)$ is a mapping  from $X \supseteq \La(\CS)$ to~$W$, i.e., from the labels of $\CS$ to the worlds of $\M$, such that  $\Delta \in [\Gamma]_a$ implies $\labint{\ell(\Gamma)}R_a \labint{\ell(\Delta)}$.  $\M, \seqint \vDash \CS$ if{f} $\M, \labint{\ell(\Gamma)} \vDash \phi$ for some component~$\Gamma$ of $\CS$ such that  $\phi\in \Gamma$. $\CS$ is \emph{valid}, written $\vDash_{\Sfive} \CS$ if{f} $\M, \seqint \vDash \CS$ for all $\M$ and $\seqint$.
\end{definition}

It is an easy exercise to show that the direct interpretation into models and formula interpretation yield the same notion of validity:
\begin{lemma}[Correctness]
$\M, w \nvDash \iota(\CS)$ if{f} there is an interpretation $\seqint$ of $\CS_\ell$ into $\M$ such that $\labint{0}=w$ and $\M,\seqint \nvDash \CS$, where $0=\ell(r(\CS))$ is the label of the root. Consequently, $\vDash_{\Sfive} \CS$ if{f}
$\iota(\CS) \in \Sfive$.
\end{lemma}

\begin{lemma}[Soundness]
\label{lem:soundess}
All rules of $\CSfive$ are sound, i.e., preserve the cross-sequent validity.
$
 \vdash_{\CSfive} \CS 
\Rightarrow
 \vDash_{\Sfive_n} \CS
$.
Consequently, also $
 \vdash_{\CSfive} \CS 
\Rightarrow
 \iota(\CS) \in \Sfive_n 
$.
\end{lemma}

Calculus $\CSfive$ is not terminating in the naive sense: rule $\Box \hasbracket a$ can be repeatedly used to create more and more components within the same $a$-bracket. However, there is a standard way of restricting rule applications to avoid duplications, and that produces a strongly terminating calculus.

\begin{definition}[Saturation]
\label{def:saturation}
    A formula $\phi$ in a component~$\Gamma$ of a cross-sequent $\CS$ is \emph{saturated} if{f} the following type-dependent condition is satisfied:
    \begin{compactitem}
        \item $\top$ is never saturated, $\bot$ is always saturated;
        \item a literal is saturated if{f} the opposite literal is not present in $\Gamma$, i.e., $p\in \Gamma$ is saturated if{f} $\notp\notin\Gamma$, and $\notp\in \Gamma$ is saturated if{f} $p\notin \Gamma$;
        \item $\phi \vee \psi\in \Gamma$ is saturated if{f}  $\{\phi,\psi\} \subseteq \Gamma$;
        \item $\phi \wedge \psi\in \Gamma$ is saturated if{f} $\{\phi,\psi\}\cap\Gamma\ne\varnothing$;
        \item $\Box_a \phi\in \Gamma$ is saturated if{f}  $\phi\in \Delta$  for some  component $\Delta \in [\Gamma]_a$;
        \item $\Diamond_a \phi\in \Gamma$ is saturated if{f}  $\phi\in \Delta$  for each  component $\Delta \in [\Gamma]_a$, i.e., $\phi \in \Gamma$, and $\phi$ is present in each $a$-parent, $a$-sibling, and $a$-child of $\Gamma$.
     \end{compactitem}
     A cross-sequent $\CS$ is \emph{saturated} if all formulas in it are saturated.
\end{definition}

\begin{lemma}[Termination]
\label{lem:termination}
Let us restrict calculus $\CSfive$ so that the rules are applied in the following order:
\begin{compactenum}
\item If there is an (unsaturated) occurrence of $\top$, then apply rule $\top$ (to this occurrence).
\item If a literal is not saturated, i.e., its opposite is present in the same component, then apply rule $\id$ to this literal and its opposite.
\item If $\phi \vee \psi$ is not saturated, then apply rule $\vee$ to it.
\item If $\phi \wedge \psi$ is not saturated, then apply rule $\wedge$ to it.
\item If $\Diamond_a \phi\in\Gamma$ is not saturated for a sequent component $\Delta$, i.e., $\phi \notin \Delta$ for some $\Delta \in [\Gamma]_a$, then apply  to $\Diamond_a \phi$  that of the rules $\Diamond T a $, $\Diamond \hasparent a$, $\Diamond \in a$, or $\Diamond \hasbracket a$ which adds $\phi$ to $\Delta$.
\item If $\Box_a \phi$ is not saturated, then apply one of the rules $\Box \in a$, $\Box \hasbracket a$, or $\Box \nobracket a$ to it, whichever is applicable.
\end{compactenum}
The resulting calculus $\TCSfive$ is strongly terminating.
\end{lemma}
\begin{proof}
Since no rule removes formulas, it is clear that, once saturated, each formula remains saturated going up the proof tree. Each component is a set of formulas, hence, is bounded by the number $M$ of subformulas of the endsequent. To each occurrence of a formula other than $\Diamond_a \phi$, at most one rule can be applied on any branch of the proof tree. For formulas $\Diamond_a \phi$, the number of rule applications along a branch is bounded by the cluster size. Each component in an $a$-cluster is created by a $\Box_a \phi$ formula from the same cluster, with each $\Box_a \phi$ creating at most one component in the cluster. Hence, the size of the cluster and the number of applications of a rule to each occurrence of $\Diamond_a \phi$ is also bounded by $M$. The number of clusters of a given parent is bounded by the number $n$ of agents. Finally, the length $m$ of any chain of components $\Gamma_0, \Gamma_1,\dots,\Gamma_m$ such that for each $i=1,\dots,m$, component $\Gamma_{i}$ is an $a_i$-child of $\Gamma_{i-1}$ for some agent~$a_i$ is bounded by $d(\CS)+\fd(\CS)$. \qed
\end{proof}

Consequently, every branch of a $\TCSfive$ derivation is bounded in depth and its leaf is either derivable (by $\id$ or $\top$) or saturated.

\begin{lemma}[Countermodel construction]
\label{lem:countermodel}
If a cross-sequent $\CS$ is saturated, it is not valid.
\end{lemma}
\begin{proof}
    It is sufficient to construct a countermodel $\M\ce(\La(\CS),R,V)$, where  $(l_1,l_2)\in R_a$ if{f} $l_1=\ell(\Gamma)$ and $l_2 = \ell(\Delta)$ for some components $\Gamma$ and $\Delta$ such that $\Delta \in [\Gamma]_a$ and where  $V(p) \ce \{l\in \La(\CS) \mid l = \ell(\Gamma), p \notin \Gamma\}$. The interpretation is the identity function: $\labint{l}\ce l$ for all $l \in \La(\CS)$. It is a standard exercise to prove by induction on the formula complexity that for each $l=\ell(\Gamma)$  and each formula $\phi \in \Gamma$, we have $\M, l \nvDash \phi$. For instance, $\top$ is not present in any component because $\CS$ is saturated. If $\notp \in \Gamma$, then $p \notin \Gamma$ due to saturation, hence, $\M,l\vDash p$ and $\M,l\nvDash\notp$. Boolean connectives are standard. If $\Box_a \phi \in \Gamma$, then $\phi \in \Delta\in[\Gamma_a]$, hence $\M,\ell(\Delta) \nvDash \phi$, meaning $\M,l\nvDash\Box_a \phi$, because $l=\ell(\Gamma)R_a\ell(\Delta)$. The case of~$\Diamond_a \phi$ is dual.  Since in the constructed model each formula is false in the world/component it belongs to, the whole cross-sequent is false too. 
    \qed
\end{proof}

\begin{theorem}[Completeness]
\label{th:completeness}
    Calculi $\TCSfive$ and $\CSfive$ are sound and complete. 
\end{theorem}
\begin{proof}
    $\TCSfive$ is a subsystem of $\CSfive$, hence it is sufficient to show the soundness of the latter and the completeness of the former.
    Soundness of $\CSfive$ was proved in \cref{lem:soundess}. To prove completeness of $\TCSfive$, assume $\TCSfive\nvDash \CS$ and non-deterministically apply rules of $\TCSfive$ until no rule is applicable. This process will terminates by \cref{lem:termination}. Since $\CS$ is not derivable,  at least one leaf $\CS_s$ is saturated. By \cref{lem:countermodel} there exists a countermodel $\M$ and interpretation $\seqint$ of $\CS_s$ into $\M$ such that $\M, \seqint \nvDash \CS_s$. Since all the formulas in the endsequent are present in the leaf, the same $\M,\seqint \nvDash \CS$, hence, $\nvDash_{\Sfive_n} \CS$. \qed
\end{proof}

\begin{corollary}[Decidability and FMP]
Logic $\Sfive_n$ is decidable and has the finite model property (FMP).
\end{corollary}
\begin{proof}
Since the proof search is strongly terminating by \cref{lem:termination}, calculus~$\TCSfive$ provides a decision algorithm. Since the countermodel is extracted from a leaf of the failed proof search, the number of worlds in it is the number of sequent components in this leaf, i.e., is finite. \qed
\end{proof}

\begin{corollary}
    Weakening and contraction, as well as cut are admissible in~$\CSfive$. Moreover, all the rules are invertible.
\end{corollary}
\begin{proof}
For cut and formula-level weakening and contraction, this directly follows from \cref{th:completeness}. The same is true for weakening/contracting components within an existing bracket. Contracting a whole bracket cannot happen since having two $a$-brackets would violate  consolidation condition~\eqref{cond:noasiblings}. Weakening by a whole bracket is admissible by \cref{th:completeness} as long as the new bracket does not violate the consolidation conditions. Since all inverted rules are weakening by formula(s), or by a component within a bracket (e.g., $\Box \in a$), or by a permitted bracket (e.g., $\Box \nobracket a$), they are also admissible.\qed
\end{proof}
\section{Lyndon Interpolation Property for $\Sfive_n$}
\label{sect:craig}

We now formulate and prove Lyndon interpolation for $\CSfive$. Since cross-sequents can easily be embedded into labelled sequents~\cite{Kuznets18}, the general method of computing interpolants for the latter can be applied to cross-sequents as is. Accordingly, we omit most of the standard proofs regarding the multicomponent extention of Maehara's interpolation method. Instead, we will focus on the additional task of proving interpolation for the agents/modalities. 
\begin{definition}[Interpolation Properties]
    A logic $\Lgk$ has an \emph{interpolation property} if{f} the following holds. If $\phi\to \psi \in \Lgk$, then there exists an \emph{interpolant} $\delta$ such that $\phi\to \delta \in \Lgk$ and $\delta\to \psi \in \Lgk$ and, additionally,  $\delta$ is constructed using only the syntactic elements common to $\phi$ and $\psi$. Different interpolation properties are obtained by varying the \emph{common syntactic elements}. The requirements are that: for the \emph{Craig interpolation property (\prp{CIP})}  $\Var(\delta)\subseteq\Var(\phi)\cap\Var(\psi)$~\cite{Craig57b}; for the \emph{Lyndon interpolation property (\prp{LIP})}  $\Varm(\delta)\subseteq\Varm(\phi)\cap\Varm(\psi)$ for each $\ast \in \{+,-\}$~\cite{Lyndon59}. We also consider the less well studied \emph{Agent Lyndon Interpolation Property \prp{ALIP}} that, in addition to the common language used in~\prp{LIP} requires that, $\Ag(\delta)\subseteq\Ag(\phi)\cap\Ag(\psi)$. Deductive versions of these properties are obtained by replacing $\to$ by $\vdash$, where $\vdash$ is some consequence relation.
\end{definition}
It directly follows that \prp{ALIP} implies \prp{LIP}, which  implies \prp{CIP}. We will show that \prp{ALIP}, the strongest of the three properties, holds for $\Sfive_n$.

In Maehara method~\cite{Maehara61,TroelstraS00} and its derivatives, interpolants are constructed by induction on the depth of a given derivation, where each sequent is \emph{split} into the left and right sides, ultimately corresponding to $\phi$ an $\psi$ of $\phi \to \psi$ respectively. In a one-sided sequent calculus, the only sequent is split by semicolon ;. For cross-sequents, every component  must be split this way, with the label of the component given as a superscript to ; when needed. Vice versa, given a split cross-sequent $\CS$, its left (right) side $L(\CS)$ ($R(\CS)$) is obtained by removing from each component the semicolon along with all formulas to the right (left) side of the semicolon.
Each rule $\rsf$ in \cref{fig:S5n} produces several split versions depending on whether the principal formula(s) is (are) on the left or right of the split. E.g., each rule $\rsf$ with one principal formula produces two versions $L\rsf$ and $R\rsf$ based on the following rule: the active formula(s) must be on the same side of the split as the principal formula, while all other formulas, including the principal one, are on the same side in the premise(s) as in the conclusion.  
\begin{definition}[Multiformula]
\label{def:multiformula}
Let $L$ be a fixed countable set of  labels (used to label components in cross-sequents). A multiformula is defined by grammar $\mho \cce \labe{l}{\phi} \mid (\mho \spconj \mho) \mid (\mho \spdisj \mho)$. $\La(\mho)$ denotes all labels used in $\mho$. Given an interpretation $\seqint$ from $X\supseteq \La(\mho)$ to worlds of a model $\M$, we define $\M, \seqint \vDash \labe{l}{\phi}$ if{f} $\M, \labint{l} \vDash \phi$;  $\M, \seqint \vDash \mho_1 \spconj \mho_2$ if{f} $\M, \seqint 
\vDash \mho_i$ for both $i=1,2$; $\M, \seqint \vDash \mho_1 \spdisj \mho_2$ if{f} $\M, \seqint 
\vDash \mho_i$ for at least one $i=1,2$. $\Var(\mho)$ and $\Ag(\mho)$, as well as their polarized versions, are defined through the label-forgetting projection $\bullet$ defined by $(\labe{l}{\phi})^\bullet \ce \phi$, $(\mho_1 \circledast \mho_2)^\bullet \ce \mho_1^\bullet \ast \mho_2^\bullet$ for $\ast \in \{\wedge, \vee\}$. E.g., $\Ag(\mho)\ce \Ag(\mho^\bullet)$.
\end{definition}

\begin{definition}[Sequent interpolation properties]
    \label{def:seqint}
    A multiformula $\mho$ is called a \emph{\prp{SLIP}  interpolant} of a split cross-sequent $\CS$ if{f} 0.~$\La(\mho) \subseteq \La(\CS)$ and 
    \begin{compactenum}
        \item\label{seqint:leftside} $\M, \seqint \nvDash \mho \Longrightarrow \M, \seqint \vDash L(\CS)$ for any interpretation $\seqint$ of $\CS$ into a model~$\M$;
        \item\label{seqint:rightside} $\M, \seqint \vDash \mho\Longrightarrow\M, \seqint \vDash R(\CS)$ for any interpretation $\seqint$ of $\CS$ into $\M$;
        \item\label{cond:LIP}
                $\Varp(\mho) \subseteq \Varn(L(\CS)) \cap \Varp(R(\CS))$,  $\Varn(\mho) \subseteq \Varp(L(\CS)) \cap \Varn(R(\CS))$.
    \end{compactenum}                  
In this case, we write $\LIPinterp{\mho}{\CS}$. If, additionally,               
\begin{compactenum}
\setcounter{enumi}{3}
\item\label{cond:ALIP}
  $\Ag(\mho) \subseteq \Ag(L(\CS)) \cap \Ag(R(\CS))$,
\end{compactenum}
then $\mho$ is called an \emph{\prp{SALIP} interpolant}, written $\ALIPinterp{\mho}{\CS}$.
\end{definition}

Let us recall types of split rules and their corresponding interpolant transformations from~\cite{Kuznets18}, adapting the general definitions to the specific case of cross-sequents for multi-agent logic and adding the conditions for common agents. The following notation is used for changing the interpretation of one label to a world~$v$: $\seqint_l^v \ce (\seqint \setminus (\{l\} \times L) \cup\{(l,v)\}$.
\begin{definition}[Split rule types] 
\label{def:splitrules}
Let $\rsf_1$ ($\rsf_2$) be a unary (binary) split cross-sequent rule inferring $\CS$ from $\CS_p$  (from $\CS_1$ and $\CS_2$). Let $\M$ be a  model,  $\seqint$~any interpretation of $\CS$ into $\M$, $i$ and $j$ range over $\{1,2\}$, and $\ast$ range over $\{+,-\}$. 
\begin{compactitem}
    \item
    Rule $\rsf_1$ is called \emph{local} if{f}   for all $\ast$, $\M$, and $\seqint$,
    \[
    \begin{array}{c@{\qquad}c}
    \Varm(L(\CS_p)) \subseteq \Varm(L(\CS)) 
    & 
    \Varm(R(\CS_p)) \subseteq \Varm(R(\CS))
    \\
    \Ag(L(\CS_p)) \subseteq \Ag(L(\CS))
    & 
    \Ag(R(\CS_p)) \subseteq \Ag(R(\CS))
    \\
    \M, \seqint \vDash L(\CS_p) 
    \Longrightarrow
    \M, \seqint \vDash L(\CS)
    &
    \M, \seqint \vDash R(\CS_p) 
    \Longrightarrow
    \M, \seqint \vDash R(\CS)
    \end{array}
    \]
    \item
    Rule $\rsf_2$ is called \emph{conjunctive} if{f},  for all $i\in\{1,2\}$, $\ast$, $\M$, and $\seqint$,
        \[
    \begin{array}{c@{\qquad}c}
    \Varm(L(\CS_i)) \subseteq \Varm(L(\CS)) 
    & 
    \Varm(R(\CS_i)) \subseteq \Varm(R(\CS))
    \\
    \Ag(L(\CS_i)) \subseteq \Ag(L(\CS))
    & 
    \Ag(R(\CS_i)) \subseteq \Ag(R(\CS))
    \\
    (\exists j)\M, \seqint \vDash L(\CS_j)\Longrightarrow\M, \seqint \vDash L(\CS)
    &
    (\forall j)\M, \seqint \vDash R(\CS_j)\Longrightarrow\M, \seqint \vDash R(\CS)
    \end{array}
    \]
    \item
    Rule $\rsf_2$ is called \emph{disjunctive} if{f},   for all $i\in\{1,2\}$, $\ast$, $\M$, and $\seqint$,
        \[
    \begin{array}{c@{\qquad}c}
    \Varm(L(\CS_i)) \subseteq \Varm(L(\CS)) 
    & 
    \Varm(R(\CS_i)) \subseteq \Varm(R(\CS))
    \\
    \Ag(L(\CS_i)) \subseteq \Ag(L(\CS))
    & 
    \Ag(R(\CS_i)) \subseteq \Ag(R(\CS))
    \\
    (\forall j)\M, \seqint \vDash L(\CS_j)\Longrightarrow\M, \seqint \vDash L(\CS)
    &
    (\exists j)\M, \seqint \vDash R(\CS_j)\Longrightarrow\M, \seqint \vDash R(\CS)
    \end{array}
    \]
    \item
    Rule $\rsf_1$ is called \emph{$\Box$-like  for $k\to_a l$}, where $k\ne l \in \La(\CS_p)$ and $a \in A$, if{f} 
    $\La(\CS_p)\setminus \La(\CS)=\{l\}$ and,      for all $\ast$, $\M$, and $\seqint$,
        \[
    \begin{array}{c@{\qquad}c}
    \Varm(L(\CS_p)) \subseteq \Varm(L(\CS)) 
    & 
    \Varm(R(\CS_p)) \subseteq \Varm(R(\CS))
    \\
    \Ag(L(\CS_p)) \subseteq \Ag(L(\CS))
    & 
    \Ag(R(\CS_p)) \subseteq \Ag(R(\CS))
    \\
    \multicolumn{2}{c}{(\exists v \in R_a(\labint{k}))\M, \seqint_l^v \vDash L(\CS_p)\quad\Longrightarrow\quad\M, \seqint \vDash L(\CS)\qquad\quad\strut}
    \\
    \multicolumn{2}{c}{(\forall v \in R_a(\labint{k}))\M, \seqint_l^v \vDash R(\CS_p)\quad\Longrightarrow\quad\M, \seqint \vDash R(\CS)\qquad\quad\strut}
    \end{array}
    \]
  \item
    Rule $\rsf_1$ is called \emph{$\Diamond$-like  for $k\to_a l$}, where $k\ne l \in \La(\CS_p)$ and $a \in A$, if{f} 
    $\La(\CS_p)\setminus \La(\CS)=\{l\}$ and,      for all $\ast$, $\M$, and $\seqint$,
        \[
    \begin{array}{c@{\qquad}c}
    \Varm(L(\CS_p)) \subseteq \Varm(L(\CS)) 
    & 
    \Varm(R(\CS_p)) \subseteq \Varm(R(\CS))
    \\
    \Ag(L(\CS_p)) \subseteq \Ag(L(\CS))
    & 
    \Ag(R(\CS_p)) \subseteq \Ag(R(\CS))
    \\
    \multicolumn{2}{c}{(\forall v \in R_a(\labint{k}))\M, \seqint_l^v \vDash L(\CS_p)\quad\Longrightarrow\quad\M, \seqint \vDash L(\CS)\qquad\quad\strut}
    \\
    \multicolumn{2}{c}{(\exists v \in R_a(\labint{k}))\M, \seqint_l^v \vDash R(\CS_p)\quad\Longrightarrow\quad\M, \seqint \vDash R(\CS)\qquad\quad\strut}
    \end{array}
    \]
 \end{compactitem}
\end{definition}

\begin{figure}[t]
\centering{
\fbox{\parbox{0.98\textwidth}{
\[
\vlinf{LR\id}{}{\CS_\ell\vlfill{p,\Gamma;^l\Delta, \notp}}{}
\qquad
\vlinf{RL\id}{}{\CS_\ell\vlfill{\notp, \Gamma;^l\Delta,p}}{}
\qquad
\vlinf{R\top}{}{\CS_\ell\vlfill{\Gamma;^l\Delta,\top}}{}
\]
\[
\vlinf{RR\id}{}{\CS_\ell\vlfill{ \Gamma;^l\Delta,p,\notp}}{}
\qquad
\vlinf{LL\id}{}{\CS_\ell\vlfill{ p,\notp,\Gamma;^l\Delta}}{}
\qquad
\vlinf{L\top}{}{\CS_\ell\vlfill{\top,\Gamma;^l\Delta}}{}
\]
\[
\vlinf{L\vee}{}{\CS_\ell\vlfill{\phi\vee\psi, \Gamma;^l\Delta}}{\CS_\ell\vlfill{\phi\vee\psi,\phi,\psi, \Gamma;^l\Delta}}{} 
\qquad
\vliinf{L\wedge}{}{\CS_\ell\vlfill{\phi\wedge\psi,\Gamma;^l\Delta}}{\CS_\ell\vlfill{\phi\wedge\psi,\phi,\Gamma;^l\Delta}}{\CS_\ell\vlfill{\phi\wedge\psi,\psi,\Gamma;^l\Delta}}{} 
\]
\[
\vlinf{R\vee}{}{\CS_\ell\vlfill{ \Gamma;^l\Delta,\phi\vee\psi}}{\CS_\ell\vlfill{ \Gamma;^l\Delta,\phi\vee\psi,\phi,\psi}}{}
\qquad
\vliinf{R\wedge}{}{\CS_\ell\vlfill{\Gamma;^l\Delta,\phi\wedge\psi}}{\CS_\ell\vlfill{\Gamma;^l\Delta, \phi\wedge\psi,\phi}}{\CS_\ell\vlfill{\Gamma;^l\Delta, \phi\wedge\psi,\psi}}{} 
\]
\[
\vlinf{L\Box \in a}{}{\CS_\ell\vlfill{[\Box_a \phi,\Gamma;^k\Delta, \onlyb|H]_a}}{\CS_\ell\vlfill{[\phi;^l|\Box_a \phi,\Gamma;^k\Delta, \onlyb|H]_a}}{} 
\qquad
\vlinf{R\Box \in a}{}{\CS_\ell\vlfill{[\Gamma;^k\Delta, \Box_a \phi, \onlyb|H]_a}}{\CS_\ell\vlfill{[;^l\phi|\Gamma;^k\Delta, \Box_a \phi,\onlyb|H]_a}}{} 
\]
\[
\vlinf{L\Box \hasbracket a}{}{\CS_\ell\vlfill{\Box_a \phi, \Gamma;^k\Delta,[H]_a, \onlyb}}{\CS_\ell\vlfill{\Box_a \phi, \Gamma;^k\Delta, [\phi;^l|H]_a, \onlyb}}{} 
\qquad
\vlinf{R\Box \hasbracket a}{}{\CS_\ell\vlfill{ \Gamma;^k\Delta,\Box_a \phi,[H]_a, \onlyb}}{\CS_\ell\vlfill{ \Gamma;^k\Delta, \Box_a \phi,[;^l\phi|H]_a, \onlyb}}{} 
\]
\[
\vlinf{L\Box\nobracket a}{\scriptscriptstyle \Box_a \phi \nobracket a}{\CS_\ell\vlfill{\Box_a \phi,\Gamma;^k\Delta}}{\CS_\ell\vlfill{\Box_a \phi,\Gamma;^k\Delta, [\phi;^l]_a}}{} 
\qquad
\vlinf{R\Box\nobracket a}{\scriptscriptstyle \Box_a \phi \nobracket a}{\CS_\ell\vlfill{\Gamma;^k\Delta,\Box_a \phi}}{\CS_\ell\vlfill{\Gamma;^k\Delta,\Box_a \phi, [;^l\phi]_a}}{} 
\]
\[
\vlinf{L\Diamond \in a}{}{\CS_\ell\vlfill{[\Diamond_a \phi,\Gamma;^k\Delta, \onlyb|\Pi;^l\Sigma,\onlyb'|H]_a}}{\CS_\ell\vlfill{[\Diamond_a \phi,\Gamma;^k\Delta, \onlyb|\phi,\Pi;^l\Sigma,\onlyb'|H]_a}}{} 
\qquad
\vlinf{R\Diamond \in a}{}{\CS_\ell\vlfill{[\Gamma;^k\Delta,\Diamond_a \phi, \onlyb|\Pi;^l\Sigma,\onlyb'|H]_a}}{\CS_\ell\vlfill{[\Gamma;^k\Delta,\Diamond_a \phi, \onlyb|\Pi;^l\Sigma,\phi,\onlyb'|H]_a}}{} 
\]
\[
\vlinf{L\Diamond \hasbracket a}{}{\CS_\ell\vlfill{\Diamond_a \phi,\Gamma;^k\Delta, \onlyb, [\Pi;^l\Sigma, \onlyb'|H]_a}}{\CS_\ell\vlfill{\Diamond_a \phi,\Gamma;^k\Delta, \onlyb, [\phi,\Pi;^l\Sigma, \onlyb'|H]_a}}{} 
\quad
\vlinf{R\Diamond \hasbracket a}{}{\CS_\ell\vlfill{\Gamma;^k\Delta,\Diamond_a \phi, \onlyb, [\Pi;^l\Sigma, \onlyb'|H]_a}}{\CS_\ell\vlfill{\Gamma;^k\Delta, \Diamond_a \phi,\onlyb, [\Pi;^l\Sigma,\phi, \onlyb'|H]_a}}{} 
\]
\[
\vlinf{L\Diamond T a}{}{\CS_\ell\vlfill{\Diamond_a \phi,\Gamma;^l\Delta}}{\CS_\ell\vlfill{\Diamond_a \phi,\phi,\Gamma;^l\Delta}}{} 
\qquad
\vlinf{R\Diamond T a}{}{\CS_\ell\vlfill{\Gamma;^l\Delta, \Diamond_a \phi}}{\CS_\ell\vlfill{\Gamma;^l\Delta,\Diamond_a \phi,\phi}}{} 
\]
\[
\vlinf{L\Diamond \hasparent a}{}{\CS_\ell\vlfill{\Pi;^l\Sigma,\onlyb,[\Diamond_a \phi,\Gamma;^k\Delta,\onlyb'|H]_a}}{\CS_\ell\vlfill{\phi,\Pi;^l\Sigma,\onlyb,[\Diamond_a \phi,\Gamma;^k\Delta,\onlyb'|H]_a}}{} 
\quad
\vlinf{R\Diamond \hasparent a}{}{\CS_\ell\vlfill{\Pi;^l\Sigma,\onlyb,[\Gamma;^k\Delta,\Diamond_a \phi,\onlyb'|H]_a}}{\CS_\ell\vlfill{\Pi;^l\Sigma,\phi,\onlyb,[\Gamma;^k\Delta,\Diamond_a \phi,\onlyb'|H]_a}}{} 
\]
}}}
\caption{Split calculus $\CSfive$ with labels. Here $\onlyb$ and $\onlyb'$ represent any number (possibly zero) of brackets indexed with pairwise distinct agents such that the resulting cross-sequent is proper.}
\label{fig:S5nsplit}
\end{figure}

\begin{lemma}[Transformations]
\label{lem:transformations}
Let rules $\rsf_1$ and $\rsf_2$ be as in \cref{def:splitrules}. Let \prp{IP} stand for either \prp{SLIP} or \prp{SALIP} interpolant.
\begin{compactenum}
    \item\label{case:local} If $\rsf_1$ is a local rule and 
    $\interp{\mho}{\CS_p}$, then $\interp{\mho}{\CS}$
    \item\label{case:conjunctive} If $\rsf_2$ is a conjunctive rule and 
    $\interp{\mho_i}{\CS_i}$ for $i=1,2$, then $\interp{\mho_1 \spconj \mho_2}{\CS}$.
    \item\label{case:disjunctive} If $\rsf_2$ is a disjunctive rule and 
    $\interp{\mho_i}{\CS_i}$ for $i=1,2$, then $\interp{\mho_1 \spdisj \mho_2}{\CS}$.
    \item\label{case:boxlike} If $\rsf_1$ is a $\Box$-like rule for $k \to_a l$ and 
    $\LIPinterp{\bigspconj\nolimits_{i=1}^m (\labe{l}{\phi_i} \spdisj \mho_i)}{\CS_p}$, where $l$ does not occur in any of $\mho_i$, then $\LIPinterp{\bigspconj\nolimits_{i=1}^m (\labe{k}{\Box_a\phi_i} \spdisj \mho_i)}{\CS}$.
   \item\label{case:diamondlike} If $\rsf_1$ is a $\Diamond$-like rule for $k \to_a l$ and 
    $\LIPinterp{\bigspdisj\nolimits_{i=1}^m (\labe{l}{\phi_i} \spconj \mho_i)}{\CS_p}$, where $l$ does not occur in any of $\mho_i$, then $\LIPinterp{\bigspdisj\nolimits_{i=1}^m (\labe{k}{\Diamond_a\phi_i} \spconj \mho_i)}{\CS}$.
\end{compactenum}
\end{lemma}
\begin{proof}
    The proof goes along the lines of the proofs of Lemmas 4.18, 4.21, 4.24, and 4.27 for local, conjunctive,  disjunctive, and $\Box$/$\Diamond$-like rules respectively in~\cite{Kuznets18}. The only non-trivial difference is that for the \prp{ALIP} case for \ref{case:local}.--\ref{case:disjunctive}. there is an additional, easily checkable condition on agents and  that in~\cite{Kuznets18} cases for \ref{case:boxlike}.--\ref{case:diamondlike}. had $\mho_i$ spelled out as a conjunction/disjunction without $l$. For space considerations, we only give the argument for $\Box$-like rules. The label conditions 0. and  atom conditions~\ref{cond:LIP}. from \cref{def:seqint} are obvious. For the left side condition~\ref{seqint:leftside}., let $\M,\seqint \nvDash \bigspconj\nolimits_{i=1}^m (\labe{k}{\Box_a\phi_i} \spdisj \mho_i)$ for the conclusion interpolant. Then there is $i$ such that   $\M, \labint{k} \nvDash \Box_a\phi_i$ and $\M, \seqint \nvDash \mho_i$. Hence, $\M, v \nvDash \phi_i$ for some $v \in R_a(\labint{k})$. Since $l \notin \La(\mho_i)$, we have $\M, \seqint_l^v \nvDash \bigspconj\nolimits_{i=1}^m (\labe{l}{\phi_i} \spdisj \mho_i)$ for the premise interpolant. We get  $\M, \seqint_l^v \vDash L(\CS_p)$ by~\ref{seqint:leftside}. for the premise and $\M, \seqint \vDash L(\CS)$ by $\Box$-likeness.

    For the right side condition~\ref{seqint:rightside}., let $\M,\seqint \vDash \bigspconj\nolimits_{i=1}^m (\labe{k}{\Box_a\phi_i} \spdisj \mho_i)$ for the conclusion interpolant. Then for all $i$ either  $\M, \labint{k} \vDash \Box_a\phi_i$ or $\M, \seqint \vDash \mho_i$. Hence, for all $v \in R_a(\labint{k})$ and all $i$, either $\M, v \vDash \phi_i$ or $\M, \seqint \vDash \mho_i$. Since $l \notin \La(\mho_i)$ for any $i$, we have $\M, \seqint_l^v \vDash \bigspconj\nolimits_{i=1}^m (\labe{l}{\phi_i} \spdisj \mho_i)$ for all $v \in R_a(\labint{k})$ for the premise interpolant. We get  $\M, \seqint_l^v \vDash R(\CS_p)$ for all $v \in R_a(\labint{k})$ by~\ref{seqint:rightside}. for the premise and $\M, \seqint \vDash R(\CS)$ by $\Box$-likeness. 
    \qed
\end{proof}
\Cref{lem:transformations} provides only a partial solution for the $\Box$-like/$\Diamond$-like rules since the transformations from~\cite{Kuznets18}, while preserving \prp{SLIP}, may produce an interpolant violating \ref{cond:ALIP}.
\begin{example}
\label{ex:nosalip}
     Split cross-sequent $
     ;\Box_a p, \Diamond_a \notp$ is derivable in the split $\CSfive$, but the algorithm from~\cite{Kuznets18} produces $\labe{0}{\Box_a \top}$ as an interpolant  even though $a$-modalities are not present on the left side:
     \[
     \vlderivation{
     \vlin{R\Box\nobracket a}{}{
    \LIPinterp{\labe{0}{\Box_a\top}}{;^0\Box_a p, \Diamond_a \notp}
    }{%
            \vlin{R\Diamond\hasbracket a}{}{%
                \ALIPinterp{\labe{1}{\top}}{;^0\Box_a p, \Diamond_a \notp, [;^1p]_a}
                }%
                {%
                    \vlin{RR\id}{}{%
                        \ALIPinterp{\labe{1}{\top}}{;^0\Box_a p, \Diamond_a \notp, [;^1p,\notp]_a}
                        }%
                        {\vlhy{}}
                }
            }
    }
     \]
     The offending $a$-modality $\Box_a$ is created by the transformation for the $\Box$-like split rule $R\Box\nobracket a$ in the process of removing a formula (here $p$) from an $a$-bracket.
     Therefore, this interpolant is a \prp{SLIP} but not a \prp{SALIP} interpolant.
     Note, however, that $ \Box_a \top \leftrightarrow \top \in \Sfive_n$ and $\Box_a \bot \leftrightarrow \bot\in \Sfive_n$. Thus, $\ALIPinterp{\labe{0}{\top}}{;^0\Box_a p, \Diamond_a \notp}$. 
\end{example}
To show that such an additional transformation from a \prp{SLIP} to a  \prp{SALIP} interpolant always exists, we need an auxiliary lemma that can be  proved by induction on the split derivation:
     \begin{lemma}
     \label{lem:aux}
         For every split cross-sequent $\CS'\vlfill{[\CS''| H]_a}$  in a split  proof of a $\CS$ with $d(\CS)=0$ and for every formula $\phi$ in $O(\CS'')$ with $O \in \{R,L\}$, there is a subformula~$\heartsuit_a \psi$ in~$O(\CS)$  for some $\heartsuit \in \{\Box,\Diamond\}$ such that $\phi$ is a subformula of $\psi$. 
     \end{lemma}
\begin{theorem}[\prp{SALIP} for $\CSfive$]
    If $\vdash_{\CSfive} \CS$ and $d(\CS)=0$, then very split of~$\CS$ has an \prp{SALIP}  interpolant.
\end{theorem}
\begin{proof}
    We start by proving that every derivable cross-sequent has a \prp{SLIP}  interpolant.
    As usual, the proof is by induction on the split $\CSfive$. We assume that all components of sequents in the proof tree are labeled coherently, i.e., satisfying conditions (A) and (B) from \cref{def:proof_label}. It is easy to assign interpolants to the axiom leaves. In particular,  $\ALIPinterp{\labe{l}{\top}}{\CS_\ell\vlfill{\Gamma;^l\Delta,\top}}$ (resp. $\ALIPinterp{\labe{l}{\bot}}{\CS_\ell\vlfill{\top,\Gamma;^l\Delta}}$) for the two split versions of $\id$
     since both $\labe{l}{\top}$ and $R(\CS\vlfill{\Gamma;\Delta,\top})=R(\CS)\vlfill{\Delta,\top}$ are always true (resp. $\labe{l}{\bot}$ is always false while $L(\CS\vlfill{\top,\Gamma;\Delta})=L(\CS)\vlfill{\top,\Gamma}$ is always true). For the same reason, $\ALIPinterp{\labe{l}{\top}}{\CS_\ell\vlfill{\Gamma;^l\Delta,p,\notp}}$ and $\ALIPinterp{\labe{l}{\bot}}{\CS_\ell\vlfill{p,\notp,\Gamma;^l\Delta}}$ for two of the four split versions of rule $\id$.
     Further, $\ALIPinterp{\labe{l}{p}}{\CS_\ell\vlfill{\notp,\Gamma;^l\Delta,p}}$ and $\ALIPinterp{\labe{l}{\notp}}{\CS_\ell\vlfill{p,\Gamma;^l\Delta,\notp}}$ for the remaining two splits of $\id$.

     For
      all unary and binary rules of $\CSfive$, the interpolant transformations from \cref{lem:transformations} can be used to construct  a \prp{SLIP}  interpolant for the endsequent. For the use of \ref{case:boxlike}.--\ref{case:diamondlike}.,  it is necessary to transform an interpolant into the requisite format, which is always possible using boolean reasoning akin to the transformation into the CNF/DNF. 

     To find a \prp{SALIP} interpolant, we use the assumption that the endsequent is of depth $0$, hence, contains no brackets. We run the algorithm for \prp{LIP} interpolants, with the following modification for every $\Box$-like or $\Diamond$-like rule $\rsf_1$ that creates  an additional $a$-modality: the premise of this rule must be of the form $\CS'\vlfill{[\CS''|H]_a}$ (with a possibly empty $H$), where the root of $\CS''$ is assigned label $l$.
     \begin{compactenum}
         \item If no~$\phi_i$ from \ref{case:boxlike}. or \ref{case:diamondlike}. (as appropriate) contains literals, then each of them is logically equivalent to either $\top$ or $\bot$ in $\Sfive_n$. Since, in addition to the equivalences mentioned in \cref{ex:nosalip}, $\Diamond_a \top \leftrightarrow \top\in \Sfive_n$ and $\Diamond_a \bot \leftrightarrow \bot\in \Sfive_n$, it is possible to replace each instance of  $\labe{k}{\Box_a \phi_i}$ or $\labe{k}{\Diamond_a \phi_i}$ with a logically equivalent $\labe{k}{\top}$ or $\labe{k}{\bot}$ that would add neither new atoms nor new agents. 
         \item Otherwise, at least one $\phi_i$ contains a literal with label $l$, which can only be created in a leaf of the proof tree for one of the two split rules $RL\id$ or $LR\id$. In either case, this leaf has the form $\CS^\dag\vlfill{[\CS^{\dag\dag}|H']_a}$ where the root of~$\CS^{\dag\dag}$ has the form $x, \Gamma;^l \Delta, y$ with $\{x,y\}=\{p,\notp\}$. In other words, a proof of $\CS$ contains formulas $x$ in  $L(\CS^{\dag\dag})$ and $y$ in  $R(\CS^{\dag\dag})$. By \cref{lem:aux}, there exist subformulas $\heartsuit^x_a \psi_x$ in $L(
         \CS)$ and $\heartsuit^y_a \psi_y$ in $R(\CS)$ where $\heartsuit^x,\heartsuit^y\in\{\Box,\Diamond\}$. Thus, $a \in \Ag(L(\CS))\cap\Ag(R(\CS))$ for the only agent possibly introduced to the interpolant by this transformation.
     \end{compactenum}
     In either case, given  a \prp{SALIP} interpolant for the premise, the result would be a \prp{SALIP} interpolant for the conclusion.\qed  
\end{proof}

\begin{corollary}
    $\Sfive_n$ has the \prp{ALIP}.
\end{corollary}
\begin{proof}
    Let $\phi \to \psi\in\Sfive_n$. By \cref{th:completeness}, $\vdash_{\CSfive} \neg \phi;^0 \psi$. Since this is a depth~$0$ split cross-sequent, $\ALIPinterp{\mho}{\neg \phi;^0 \psi}$ for some $\mho$ with $0$ as the only label. It is straightforward to show that $\M, \seqint \vDash \mho$ if{f} $\M, \labint{0} \vDash \mho^\bullet$ for the projection~$\bullet$ from \cref{def:multiformula}. Hence, $\mho^\bullet$ is an \prp{ALIP} interpolant of $\phi \to \psi$. Since the common language conditions are immediate, we only show that $\phi \to \mho^\bullet \in \Sfive_n$. If $\M, w \vDash \phi$, then $\M, \labint{0} \nvDash \neg \phi$ for the interpretation $\seqint\ce \{(0,w)\}$. Hence, $\M, \seqint \nvDash L(\neg\phi;^0\psi)$. By \cref{def:seqint}.\ref{seqint:leftside}, $\M, \seqint \vDash \mho$. Hence, $\M, \labint{0} 
    \vDash \mho^\bullet$, i.e., $\M, w \vDash \mho^\bullet$. The case of $\mho^\bullet \to \psi \in \Sfive_n$ is even simpler.\qed
\end{proof}

\section{Towards Distributed Knowledge}
\label{sect:distr}
In their recent paper~\cite{MuraiS20}, Murai and Sano developed cut-free sequent calculi for logics of distributed knowledge based on $\K$, {\sf KT}, {\sf KD}, and {\sf S4} modalities and used these calculi for to prove the \prp{CIP} for both atoms and agents. They left the question of $\Sfive$-based distributed knowledge open, suggesting to either use an analytic sequent calculus with cut or to use generalized cut-free sequent calculi akin to the type developed in the current paper. The former method, however  is not suitable for proving the \prp{LIP} as the cut rule does not respect the polarities. One of our goals was to use cross-sequents to prove the stronger \prp{ALIP} for $\Sfive$-based distributed knowledge. 

Unfortunately, the interaction of several distributed modalities within the hierarchy of agent groups does not allow for a direct application of cross-sequents. The natural way to use cross-sequents would be to consider each agent group to have its own modality and supply these modalities with a hierarchy of strength, so that $D_G \phi \to D_H \phi$ whenever $G \subseteq H$, where $G$ and $H$ are groups of agents and $D_G$ stands for the distributed knowledge of group $G$. Unfortunately, adding these additional connections would violate the termination property we used to prove the completeness theorem.
Consider, e.g., $\theta\ce \widehat{D}_{a}(D_{ab} p \vee D_{ac} q)$. Here $a$'s~knowledge  implies distributed knowledge of both the $ab$ group and the $ac$ group. That suggests that any any $ab$-cluster and any $ac$-cluster should be contained in some $a$-cluster. Moreover, by transitivity, any chain of $ab$- and $ac$-clusters should fully be within one $a$ cluster. Thus, the $\Diamond$-like formula $\theta$ should deposit the disjunction $D_{ab} p \vee D_{ac} q$ in every component of such a chain. But the $\Box$-like modalities $D_{ab} p$ and $D_{ac} q$ are prone to initiate the creation of new clusters, which would alternate forever. Because one $\widehat{D} \phi$ formula can now deposit~$\phi$ arbitrarily deep (further away from the root) within the cross-sequent, the proof search would not terminate the way it does for $\Sfive_n$, where $\Diamond_a \phi$ can deposit $\phi$ at most one level  deeper. The question of interpolation for $\Sfive$-based distributed knowledge, thus, remains future research. 

Apart from pursuing some variant or elaboration of cross-sequents, 
another possible direction could be to use indexed hypersequents from~\cite{Poggiolesi13}, which use a similar intuition to cross-sequents but use Aumann structures instead of Kripke models.

It may be worth noting that cross-sequents can be easily amended to handle $\Sfive_n$ with added isolated axioms of the type $\Box_a \phi\to \Box_b \phi$ as long as no two such axioms share an agent. Since this extension is not likely to lead to a result on distributed knowledge we chose not to pursue it.

\section{Deductive Interpolation}
\label{sect:deductive interpolation}

\Cref{sect:craig} establishes \prp{LIP} and, hence, \prp{CIP} for $\Sfive_n$. For any normal modal logic $\Lgk$, \prp{CIP} for $\Lgk$ is tantamount to interpolation for its local consequence relation. In this section, we change gears and consider interpolation for the \emph{global} consequence relations of the logics we have previously considered.

The following local deduction theorem gives a crucial link between local and global consequence. It is folklore to a great extent, but a proof may be found, for example, in \cite[Theorem 3]{FZ2021}. To state the theorem precisely, we first give a key definition: By a \emph{block} we mean any word over the alphabet $\{\Box_a \mid a\in A\}$.

\begin{proposition}\label{prop:DDT}
    Let $\Gamma\cup\{\phi,\psi\}$ be any set of formulas and $\Lgk$ be any normal logic. Then $\Gamma,\phi\vdash_\Lgk\psi$ if and only if there exists an integer $m\geq 0$ and blocks $B_1,\ldots,B_m$ such that $\Gamma\vdash_\Lgk (\phi\wedge B_1\phi\wedge\ldots\wedge B_m\phi)\to\psi$.
\end{proposition}

It is well known that local deduction theorems of this kind typically entail that the \prp{CIP} implies the \emph{deductive interpolation property} (\prp{DIP}), but this is only true for \emph{propositional variables}. For agents, the situation is more complicated. 

\begin{proposition}
    Suppose $\Lgk$ is any normal normal logic with \prp{CIP} for both propositional variables and agents. Then $\vdash_\Lgk$ has the following restricted form of \prp{DIP}: Whenever $\psi\vdash\phi$, there exists a formula $\delta$ such that $\prop(\delta)\subseteq\prop(\psi)\cap\prop(\phi)$ and $\agt(\delta)\subseteq\agt(\psi)$.
\end{proposition}

\begin{proof}
    Suppose that $\phi\vdash\psi$. Applying \cref{prop:DDT}, there exists $n\geq 0$ and blocks $B_1,\ldots,B_m$ such that $\vdash_\Lgk(\phi\wedge B_1\phi\wedge\ldots\wedge B_m\phi)\to\psi$. By \prp{CIP}, we obtain a formula~$\delta$ such that $\vdash_\Lgk (\phi\wedge B_1\phi\wedge\ldots\wedge B_m\phi)\to\delta$, $\vdash_\Lgk\delta\to\psi$, and 
    \begin{gather*}
    \prop(\delta)\subseteq\prop(\phi\wedge B_1\phi\wedge\ldots\wedge B_m\phi)\cap\prop(\psi)=\prop(\phi)\cap\prop(\psi),
    \\
    \agt(\delta)\subseteq\agt(\phi\wedge B_1\phi\wedge\ldots\wedge B_m\phi)\cap\agt(\psi)\subseteq\agt(\psi). 
    \end{gather*}
\end{proof}

As the foregoing proof illustrates, the introduction of additional agents in the left-to-right direction of the local deduction theorem precludes deriving agent interpolation for the global consequence from agent interpolation for the local consequence by the usual argument. 

\section{Conclusion and Future Work}
We have developed a novel proof formalism for multi-agent $\Sfive_n$ and used it to prove the Lyndon interpolation property for propositional atoms simultaneously with the (Craig) interpolation property for agents. It is interesting to note that proof-theoretic method does not natively support agent interpolation and requires a bit of processing.

Unfortunately, similar to previous attempts, this calculus does not directly lead to distributed knowledge based on $\Sfive$. Apart from adapting cross-sequents to distributed knowledge, one can attempt using more general labelled sequents~\cite{Negri05,NegrivP11}, although it is unlikely that termination problem there would be simpler to solve than for the more structured cross-sequents. Perhaps, the interconnected distributed knowledge modalities require their own special sequent structure reflecting their complexity.
Another natural direction of future research is to extend this formalism to weaker extensions of multi-agent {\sf K5}, all of which share the cluster-type semantics with~$\Sfive$.

Although the focus of this paper has been on interpolation, the proof-theoretic formalism introduced here is also relevant to complexity. This will be the subject of future work.

\bibliographystyle{abbrvurl}
\bibliography{tableaux25}
\end{document}